\documentclass{article}
\usepackage{amssymb}
\usepackage{amsfonts}
\usepackage{amsmath}

\newtheorem{theorem}{Theorem}

\newtheorem{corollary}[theorem]{Corollary}

\newtheorem{example}[theorem]{Example}

\newtheorem{lemma}[theorem]{Lemma}

\newtheorem{proposition}[theorem]{Proposition}

\newenvironment{proof}[1][Proof]{\noindent\textbf{#1.} }{\ \rule{0.5em}{0.5em}}
\input{tcilatex}

\begin{document}

\title{On Quantum Codes Obtained From Cyclic Codes Over $%
F_{2}+vF_{2}+v^{2}F_{2}$}
\author{Abdullah Dertli$^{a}$, Yasemin Cengellenmis$^{b}$, Senol Eren$^{a}$
\and $\left( a\right) $ Ondokuz May\i s University, Faculty of Arts and
Sciences, \and Mathematics Department, Samsun, Turkey \and %
abdullah.dertli@gmail.com, seren@omu.edu.tr \and $\left( b\right) $ Trakya
University, Faculty of Arts and Sciences, \and Mathematics Department,
Edirne, Turkey \and ycengellenmis@gmail.com}
\maketitle

\begin{abstract}
A new Gray map which is both an isometry and a weight preserving map from $%
R=F_{2}+vF_{2}+v^{2}F_{2}$ to $F_{2}^{3}$ is defined. A construction for
quantum error correcting codes from cyclic codes over finite ring $%
R=F_{2}+vF_{2}+v^{2}F_{2},$ $v^{3}=v$ is given.The parameters of quantum
codes which are obtained from cyclic codes over $R$ are determined.
\end{abstract}

\section{Introduction}

Quantum computer has extraordinary power and it can certain hard problems
much more quickly than any classical computer. It is a computation device
that makes direct use of quantum mechanical phenomena to perform on data.

Quantum computer inevitably interacts with its surrounding, resulting in
decoherence and hence in the decay of quantum information stored in the
device.

Quantum error correction is used in quantum computing to protect quantum
information from errors due to decoherence and other quantum noise.
Effective techniques for quantum error correction were first developed by
Shor and independently by Steane. They discovered quantum error correcting
codes.

Although the theory quantum error correcting codes has striking differences
from the theory classical error correcting codes, Calderbank et al. gave a
way to construct quantum error correcting codes from classical error
correcting codes.

Many good quantum codes have been constructed by using classical cyclic
codes over $F_{q}$ with self orthogonal (or dual containing) properties.

Some authors constructed quantum codes by using linear codes over finite
rings. For example, in $\left[ 4\right] $, J.Qian et al. gave a new method
to obtain self-orthogonal codes over $F_{2}$. They gave a construction for
quantum error correcting codes starting from cyclic codes over finite ring, $%
F_{2}+uF_{2}$, $u^{2}=0.$ X.Kai, S.Zhu gave construction for quantum codes
from linear and cyclic codes over $F_{4}+uF_{4}$, $u^{2}=0$ in $\left[ 7%
\right] .$ They derived Hermitian self-orthogonal codes over $F_{4}$ as Gray
images of linear and cyclic codes over $F_{4}+uF_{4}.$ In $\left[ 9\right] $%
, X.Yin and W.Ma gave an existence condition of quantum codes which are
derived from cyclic codes over finite ring $F_{2}+uF_{2}+u^{2}F_{2},$ $%
u^{3}=0$ with Lee metric. J.Qian gave a new method of constructing quantum
error correcting codes from cyclic codes over finite ring $F_{2}+vF_{2},$ $%
v^{2}=v$, for arbitrary length $n$\ in $\left[ 3\right] .$

This paper is organized as follows. In section $2,$ we give some basic
knowledges about the finite ring $R,$ cyclic code, dual code. In section $3$%
, we define a new Gray map from $R$ to $F_{2}^{3},$ Lee weights of elements
of $R$ and Lee distance in the linear codes over $R.$ We show that if $C$ is
self orthogonal so is $\psi \left( C\right) $ and we obtain in the Gray
images of cyclic codes over $R.$ In section $4$, a sufficient and necessary
condition for cyclic code over $R$ that contains its dual is given. The
parameters of quantum error correcting codes are obtained from cyclic codes
over $R$.

\section{Preliminaries}

Let $R$ be the ring $F_{2}+vF_{2}+v^{2}F_{2}$ where $v^{3}=v$ and $%
F_{2}=\left\{ 0,1\right\} $. $R$ is a commutative ring with eight elements. $%
R$ is semilocal ring with two maximal ideals and principal ideal ring with
characteristic $2$. It is not finite chain ring. It has six ideals. The
ideal with $8$ elements which \ generate the ring $R$ is as follows

\begin{equation*}
\left\langle 1\right\rangle =\left\langle 1+v+v^{2}\right\rangle =R
\end{equation*}

The two maximal ideals with $4$ elements are as follows

\begin{eqnarray*}
\left\langle v\right\rangle &=&\left\langle v^{2}\right\rangle =\left\{
0,v,v^{2},v+v^{2}\right\} \\
\left\langle 1+v\right\rangle &=&\left\{ 0,1+v,1+v^{2},v+v^{2}\right\}
\end{eqnarray*}

The ideals with two elements are as follows

\begin{eqnarray*}
\left\langle 1+v^{2}\right\rangle &=&\left\{ 0,1+v^{2}\right\} \\
\left\langle v+v^{2}\right\rangle &=&\left\{ 0,v+v^{2}\right\}
\end{eqnarray*}

The element $1$ and $1+v+v^{2}$ are two units of $R.$ The zero divisors are $%
0,v,v^{2},1+v,v+v^{2},1+v+v^{2}.$

A linear code $C$ over $R$ length $n$ is a $R-$submodule of $R^{n}$. An
element of $C$ is called a codeword.

For any $x=\left( x_{0},x_{1},...,x_{n-1}\right) $, $y=\left(
y_{0},y_{1},...,y_{n-1}\right) $ the inner product is defined as

\begin{equation*}
x.y=\sum_{i=0}^{n-1}x_{i}y_{i}
\end{equation*}

If $x.y=0$ then $x$ and $y$ are said to be orthogonal. Let $C$ be linear
code of length $n$ over $R$, the dual code of $C$ 
\begin{equation*}
C^{\perp }=\left\{ x:\forall y\in C,x.y=0\right\}
\end{equation*}%
which is also a linear code over $R$ of length $n$. A code $C$ is self
orthogonal if $C\subseteq C^{\perp }$ and self dual if $C=C^{\perp }$.

A cyclic code $C$ over $R$ is a linear code with the property that if $%
c=\left( c_{0},c_{1},...,c_{n-1}\right) \in C$ then $\sigma \left( C\right)
=\left( c_{n-1},c_{0},...,c_{n-2}\right) \in C.$ A subset $C$ of $R^{n}$ is
a linear cyclic code of length $n$ iff it is polynomial representation is an
ideal of $R\left[ x\right] /\left\langle x^{n}-1\right\rangle .$

Let $C$ be code over $F_{2}$ of length $n$ and $\acute{c}=\left( \acute{c}%
_{0},\acute{c}_{1},...,\acute{c}_{n-1}\right) $ be a codeword of $C.$ The
Hamming weight of $\acute{c}$ is defined as $w_{H}\left( \acute{c}\right)
=\dsum\limits_{i=0}^{n-1}w_{H}\left( \acute{c}_{i}\right) $ where $%
w_{H}\left( \acute{c}_{i}\right) =1$ if $\acute{c}_{i}=1$ and $w_{H}\left( 
\acute{c}_{i}\right) =0$ if $\acute{c}_{i}=0.$ Hamming distance of $C$ is
defined as $d_{H}\left( C\right) =\min d_{H}\left( c,\acute{c}\right) ,$
where for any $\acute{c}\in C,$ $c\neq \acute{c}$ and $d_{H}\left( c,\acute{c%
}\right) $ is Hamming distance between two codewords with $d_{H}\left( c,%
\acute{c}\right) =w_{H}\left( c-\acute{c}\right) .$

Let $a\in F_{2}^{3n}$ with $a=\left( a_{0},a_{1},...,a_{3n-1}\right) =\left(
a^{\left( 0\right) }\left\vert a^{\left( 1\right) }\right\vert a^{\left(
2\right) }\right) ,$ $a^{\left( i\right) }\in F_{2}^{n}$ for $i=0,1,2.$ Let $%
\varphi $ be a map from $F_{2}^{3n}$ to $F_{2}^{3n}$ given by $\varphi
\left( a\right) =\left( \sigma \left( a^{\left( 0\right) }\right) \left\vert
\sigma \left( a^{\left( 1\right) }\right) \right\vert \sigma \left(
a^{\left( 2\right) }\right) \right) $ where $\sigma $ is a cyclic shift from 
$F_{2}^{n}$ to $F_{2}^{n}$ given by $\sigma \left( a^{\left( i\right)
}\right) =((a^{\left( i,n-1\right) }),(a^{\left( i,0\right) }),...,$ $\ \ \
\ \ (a^{\left( i,n-2\right) }))$ for every $a^{\left( i\right) }=\left(
a^{\left( i,0\right) },...,a^{\left( i,n-1\right) }\right) $ where $%
a^{\left( i,j\right) }\in F_{2}$, $j=0,1,...,n-1.$ A code of length $3n$
over $F_{2}$ is said to be quasi cyclic code of index $3$ if $\varphi \left(
C\right) =C.$

\section{Gray Map And Gray Images Of Cyclic Codes Over $R$}

Let $x=a+vb+v^{2}c$ be an element of $R$ where $a,b,c\in F_{2}$. We define
Gray map $\psi $ from $R$ to $F_{2}^{3}$ by

\begin{eqnarray*}
\psi &:&R\rightarrow F_{2}^{3} \\
a+vb+v^{2}c &\mapsto &\psi \left( a+vb+v^{2}c\right) =\left( a,b,a+c\right)
\end{eqnarray*}

From definition, the Lee weight of elements of $R$ are defined as follows

\begin{eqnarray*}
w_{L}\left( 0\right) &=&0\text{ \ \ \ \ \ \ \ \ \ \ \ \ \ \ \ \ \ \ \ }%
w_{L}\left( 1+v\right) =3 \\
w_{L}\left( 1\right) &=&2\text{ \ \ \ \ \ \ \ \ \ \ \ \ \ \ \ \ \ \ \ }%
w_{L}\left( 1+v^{2}\right) =1 \\
w_{L}\left( v\right) &=&1\text{ \ \ \ \ \ \ \ \ \ \ \ \ \ \ \ \ \ \ \ }%
w_{L}\left( v+v^{2}\right) =2 \\
w_{L}\left( v^{2}\right) &=&1\text{ \ \ \ \ \ \ \ \ \ \ \ \ \ \ \ \ \ \ \ }%
w_{L}\left( 1+v+v^{2}\right) =2
\end{eqnarray*}

Let $C$ be a linear code over $R$ of length $n$. For any codeword $c=\left(
c_{0},...,c_{n-1}\right) $ the Lee weight of $c$ is defined as $w_{L}\left(
c\right) =\dsum\limits_{i=0}^{n-1}w_{L}\left( c_{i}\right) $ and the Lee
distance of $C$ is defined as $d_{L}\left( C\right) =\min d_{L}\left( c,%
\acute{c}\right) ,$ where for any $\acute{c}\in C,$ $c\neq \acute{c}$ and $%
d_{L}\left( c,\acute{c}\right) $ is Lee distance between two codewords with $%
d_{L}\left( c,\acute{c}\right) =w_{L}\left( c-\acute{c}\right) .$ Gray map $%
\psi $ can be extended to map from $R^{n}$ to $F_{2}^{3n}.$

\begin{theorem}
The Gray map $\psi $ is a weight preserving map from $\left( R^{n},\text{Lee
weight}\right) $ to $\left( F_{2}^{3n},\text{Hamming weight}\right) .$
Moreover it is an isometry from $R^{n}$ to $F_{2}^{3n}.$
\end{theorem}

\begin{proof}
Similar to proof of theorem $3.5$ in $\left[ 9\right] $.
\end{proof}

\begin{theorem}
If $C$ is an $\left( n,k,d_{L}\right) $ linear codes over $R$ then $\psi
\left( C\right) $ is a $\left( 3n,k,d_{H}\right) $ linear codes over $F_{2},$
where $d_{H}=d_{L}.$
\end{theorem}

\begin{proof}
Let $c=a_{1}+vb_{1}+v^{2}c_{1},$ $\acute{c}=a_{2}+vb_{2}+v^{2}c_{2}\in
R,\alpha \in F_{2\text{ }}$ then

$\psi \left( c+\acute{c}\right) =\psi \left( a_{1}+a_{2}+v\left(
b_{1}+b_{2}\right) +v^{2}\left( c_{1}+c_{2}\right) \right) $

\ \ \ \ \ \ \ \ \ \ \ \ \ \ \ \ \ $=\left(
a_{1}+a_{2},b_{1}+b_{2},a_{1}+a_{2}+c_{1}+c_{2}\right) $

\ \ \ \ \ \ \ \ \ \ \ \ \ \ \ \ \ $=\left( a_{1},b_{1},a_{1}+c_{1}\right)
+\left( a_{2},b_{2},a_{2}+c_{2}\right) $

\ \ \ \ \ \ \ \ \ \ \ \ \ \ \ \ \ $=\psi \left( c\right) +\psi \left( \acute{%
c}\right) $

$\ \ \ \ \psi \left( \alpha c\right) =\psi \left( \alpha a_{1}+v\alpha
b_{1}+v^{2}\alpha c_{1}\right) $

\ \ \ \ \ \ \ \ \ \ \ \ \ \ \ \ $=\left( \alpha a_{1},\alpha b_{1},\alpha
a_{1}+\alpha c_{1}\right) $

\ \ \ \ \ \ \ \ \ \ \ \ \ \ \ \ $=\alpha \left(
a_{1},b_{1},a_{1}+c_{1}\right) $

\ \ \ \ \ \ \ \ \ \ \ \ \ \ \ \ $=\alpha \psi \left( c\right) $

so $\psi $ is\textit{\ }linear\textit{. }As\textit{\ }$\psi $ is bijective
then $\left\vert C\right\vert =\left\vert \psi \left( C\right) \right\vert $%
. From theorem $1$ we have $d_{H}=d_{L}.$
\end{proof}

\begin{theorem}
If $C$ is self orthogonal, so is $\psi \left( C\right) .$
\end{theorem}

\begin{proof}
Let $c=a_{1}+vb_{1}+v^{2}c_{1},$ $\acute{c}=a_{2}+vb_{2}+v^{2}c_{2}$ where $%
a_{1},b_{1},c_{1},$ $a_{2},b_{2},c_{2}\in F_{2}$.

$c.\acute{c}=a_{1}a_{2}+v\left(
a_{1}b_{2}+b_{1}a_{2}+b_{1}c_{2}+c_{1}b_{2}\right) +v^{2}\left(
a_{1}c_{2}+b_{1}b_{2}+c_{1}a_{2}+c_{1}c_{2}\right) $, if$\ C$ is self
orthogonal,so we have $%
a_{1}a_{2}=0,a_{1}b_{2}+b_{1}a_{2}+b_{1}c_{2}+c_{1}b_{2}=0,a_{1}c_{2}+b_{1}b_{2}+c_{1}a_{2}+c_{1}c_{2}=0 
$. From

$\psi \left( c\right) .\psi \left( \acute{c}\right) =\left(
a_{1},b_{1},a_{1}+c_{1}\right) \left( a_{2},b_{2},a_{2}+c_{2}\right) $

\ \ \ \ \ \ \ \ \ \ \ \ \ \ \ \ \ \ \ \ \ $%
=a_{1}a_{2}+b_{1}b_{2}+a_{1}a_{2}+a_{1}c_{2}+c_{1}a_{2}+c_{1}c_{2}=0$

Therefore, we have $\psi \left( C\right) $ is self orthogonal.
\end{proof}

\begin{proposition}
Let $\psi $ the Gray map from $R^{n}$ to $F_{2}^{3n}$ as follows%
\begin{equation*}
\psi \left( r_{0},r_{1},...,r_{n-1}\right) =\left(
a_{0},...,a_{n-1},b_{0},...,b_{n-1},a_{0}+c_{0},...,a_{n-1}+c_{n-1}\right)
\end{equation*}%
where $r_{i}=a_{i}+vb_{i}+v^{2}c_{i}$ for $i=0,1,...,n-1.$ Let $\sigma $ be
cyclic shift and let $\varphi $ be a map from $F_{2}^{3n}$ to $F_{2}^{3n}$
as in the preliminaries. Then $\psi \sigma =\varphi \psi .$
\end{proposition}

\begin{proof}
Let $r_{i}=a_{i}+vb_{i}+v^{2}c_{i}$ be the elements of $R$ for $%
i=0,1,....,n-1.$ We have $\sigma \left( r_{0},r_{1},...,r_{n-1}\right)
=\left( r_{n-1},r_{0},...,r_{n-2}\right) .$ If \ we apply $\psi ,$ we have 
\begin{eqnarray*}
\psi \left( \sigma \left( r_{0},...,r_{n-1}\right) \right) &=&\psi
(r_{n-1},r_{0},...,r_{n-2}) \\
&=&(a_{n-1},...,a_{n-2},b_{n-1},...,b_{n-2},a_{n-1}+c_{n-1},...,a_{n-2}+c_{n-2})
\end{eqnarray*}%
On the other hand $\psi
(r_{0},...,r_{n-1})=(a_{0},...,a_{n-1},b_{0},...,b_{n-1},a_{0}+c_{0},...,a_{n-1}+c_{n-1}) 
$. If we apply $\varphi ,$ we have $\varphi \left( \psi \left(
r_{0},r_{1},...,r_{n-1}\right) \right) =(a_{n-1},...,a_{n-2},b_{n-1},...,$ $%
b_{n-2},a_{n-1}+c_{n-1},...,a_{n-2}+c_{n-2})$. Thus, $\psi \sigma =\varphi
\psi .$
\end{proof}

\begin{proposition}
Let $\sigma $ and $\varphi $ be as above. A code $C$ of length $n$ over $R$
is cyclic code if and only if $\psi \left( C\right) $ is quasi cyclic code
of index $3$ and length $3n$ over $F_{2}.$
\end{proposition}

\begin{proof}
Suppose $C$ is cyclic code. Then $\sigma \left( C\right) =C.$ If we apply $%
\psi ,$ we have $\psi \left( \sigma \left( C\right) \right) =\psi \left(
C\right) .$ From proposition 4, $\psi \left( \sigma \left( C\right) \right)
=\varphi \left( \psi \left( C\right) \right) =\psi \left( C\right) .$ Hence, 
$\psi \left( C\right) $ is a quasi cyclic code of index $3$. Conversely, if $%
\psi \left( C\right) $ is a quasi cyclic code of index $3$, then $\varphi
(\psi \left( C\right) )=\psi \left( C\right) .$ From proposition 4, we have $%
\varphi \left( \psi \left( C\right) \right) =\psi \left( \sigma \left(
C\right) \right) =\psi \left( C\right) .$ Since $\psi $ is injective, it
follows that $\sigma \left( C\right) =C.$
\end{proof}

We know that any linear code over $R$ of length $n$ is permutation
equivalent to a code with generator matrix of the form

\begin{equation*}
G=\left[ 
\begin{array}{cccccc}
I_{k_{1}} & A_{1} & A_{2} & A_{31}+vA_{32} & A_{41}+vA_{42} & A_{5} \\ 
0 & vI_{k_{2}} & B_{1} & B_{21}+vB_{22} & vB_{31} & B_{4} \\ 
0 & 0 & \left( 1+v\right) I_{k_{3}} & \left( 1+v\right) C_{1} & \left(
1+v\right) C_{2} & C_{3} \\ 
0 & 0 & 0 & \left( v+v^{2}\right) I_{k_{4}} & \left( v+v^{2}\right) D_{1} & D
\\ 
0 & 0 & 0 & 0 & \left( 1+v^{2}\right) I_{k_{5}} & E%
\end{array}%
\right]
\end{equation*}

where $I_{k_{1}},I_{k_{2}},I_{k_{3}},I_{k_{4}},I_{k_{5}}$ are all unit
matrices with order $k_{1},k_{2},k_{3},k_{4},k_{5}$ respectively. Let $%
k=k_{1}+k_{2}+k_{3}+k_{4}+k_{5}$ $\left[ 8\right] .$ $%
A_{5}=A_{51}+vA_{52}+v^{2}A_{53}$ , $B_{4}=vB_{41}+v^{2}B_{42}$, $%
C_{1}=C_{11}+vC_{12}+v^{2}C_{13}$, $C_{3}=\left( 1+v\right) \left(
C_{31}+vC_{32}+v^{2}C_{33}\right) $, $D=\left( v+v^{2}\right) D_{2}$, $%
E=\left( 1+v^{2}\right) E^{^{\prime }}$ where $\
A_{1},A_{2},A_{31,}A_{32},A_{41},A_{42},A_{51},A_{52},A_{53},B_{1}$, $%
B_{21},B_{22},B_{31},B_{41},B_{42},C_{11},C_{12},C_{13},C_{2},C_{31},C_{32},C_{33},D_{1},D_{2},E^{^{\prime }} 
$ are matrices over $F_{2}$ in $\left[ 8\right] .$

Let $C$ be a linear code of length $n$ over $R.$ Define

\begin{eqnarray*}
C_{1} &=&\left\{ a\in F_{2}^{n}:\exists b,c\in F_{2}^{n},a+vb+v^{2}c\in
C\right\} \\
C_{2} &=&\left\{ b\in F_{2}^{n}:\exists a,c\in F_{2}^{n},a+vb+v^{2}c\in
C\right\} \\
C_{3} &=&\left\{ a+c\in F_{2}^{n}:\exists b\in F_{2}^{n},a+vb+v^{2}c\in
C\right\}
\end{eqnarray*}

Obviously, $C_{1},C_{2}$ and $C_{3}$ are binary linear codes.

We denote that $A\otimes B\otimes C=\left\{ \left( a,b,c\right) :a\in A,b\in
B,c\in C\right\} $ and $A\oplus B\oplus C=\left\{ a+b+c:a\in A,b\in B,c\in
C\right\} .$

\begin{theorem}
Let $C$ be a linear code of length $n$ over $R.$ Then $\psi \left( C\right)
=C_{1}\otimes C_{2}\otimes C_{3}$ and $\left\vert C\right\vert =\left\vert
C_{1}\right\vert \left\vert C_{2}\right\vert \left\vert C_{3}\right\vert .$
\end{theorem}

\begin{proof}
For any $%
(a_{0},a_{1},...,a_{n-1},b_{0},b_{1},...,b_{n-1},a_{0}+c_{0},a_{1}+c_{1},...,a_{n-1}+c_{n-1})\in \psi \left( C\right) . 
$ Let $d_{i}=v\left( a_{i}+b_{i}+c_{i}\right) +\left( 1+v\right) \left(
a_{i}+b_{i}\right) +\left( 1+v^{2}\right) c_{i},$ $i=0,1,...,n-1.$ Since\ $%
\psi $ is a bijection $d=\left( d_{0},d_{1},...,d_{n-1}\right) \in C.$ By
definitions of $C_{1},C_{2}$ and $C_{3}$ we have $\left(
a_{0},a_{1},...,a_{n-1}\right) \in C_{1},$ $\left(
b_{0},b_{1},...,b_{n-1}\right) \in C_{2},$ $%
(a_{0}+c_{0},a_{1}+c_{1},...,a_{n-1}+c_{n-1})\in C_{3}$. So, $%
(a_{0},a_{1},...,a_{n-1},b_{0},b_{1},...,b_{n-1},a_{0}+c_{0},a_{1}+c_{1},...,a_{n-1}+c_{n-1})\in C_{1}\otimes C_{2}\otimes C_{3} 
$. That is $\psi \left( C\right) \subseteq C_{1}\otimes C_{2}\otimes C_{3}$.
On the other hand, for any $(a_{0},a_{1},...,a_{n-1},$ $\ \ \ \
b_{0},b_{1},...,b_{n-1},a_{0}+c_{0},a_{1}+c_{1},...,a_{n-1}+c_{n-1})\in
C_{1}\otimes C_{2}\otimes C_{3}$ where $(a_{0},a_{1},...,a_{n-1})\in C_{1},$ 
$(b_{0},b_{1},...,b_{n-1})\in C_{2}$, $%
(a_{0}+c_{0},a_{1}+c_{1},...,a_{n-1}+c_{n-1})\in C_{3}$. There are $x=\left(
a_{0},a_{1},...,a_{n-1}\right) ,$ $y=\left( b_{0},b_{1},...,b_{n-1}\right) ,$
$z=\left( c_{0},c_{1},...,c_{n-1}\right) \in C$ such that $%
x_{i}=a_{i}+vp_{i},$ $y_{i}=b_{i}+\left( 1+v\right) q_{i},$ $z_{i}=\left(
a_{i}+c_{i}\right) +\left( 1+v^{2}\right) r_{i}$ where $p_{i},$ $%
q_{i},r_{i}\in F_{2}$ and $0\leq i\leq n-1.$ Since $C$ is linear we have $%
d=vx+\left( 1+v\right) y+\left( 1+v^{2}\right) z=v\left( a+b+c\right)
+\left( 1+v\right) \left( a+c\right) +\left( 1+v^{2}\right) c\in C.$ It
follows then $\psi \left( d\right) =\left(
a_{0},a_{1},...,a_{n-1},b_{0},b_{1},...,b_{n-1},a_{0}+c_{0},a_{1}+c_{1},...,a_{n-1}+c_{n-1}\right) 
$, which gives $C_{1}\otimes C_{2}\otimes C_{3}\subseteq \psi \left(
C\right) .$ Therefore, $\psi \left( C\right) =C_{1}\otimes C_{2}\otimes
C_{3}.$ The second result is easy to verify.
\end{proof}

\begin{corollary}
If $\psi \left( C\right) =C_{1}\otimes C_{2}\otimes C_{3},$ then $%
C=vC_{1}\oplus \left( 1+v\right) C_{2}\oplus \left( 1+v^{2}\right) C_{3}.$
\end{corollary}

Moreover,

\begin{eqnarray*}
G_{1} &=&\left[ 
\begin{array}{cccccc}
I_{k_{1}} & A_{1} & A_{3} & A_{31} & A_{41} & A_{51} \\ 
0 & 0 & B_{1} & B_{21} & 0 & 0 \\ 
0 & 0 & I_{k_{3}} & C_{1} & C_{2} & C_{31} \\ 
0 & 0 & 0 & 0 & I_{k_{5}} & E^{^{\prime }}%
\end{array}%
\right] \\
G_{2} &=&\left[ 
\begin{array}{cccccc}
0 & 0 & 0 & A_{32} & A_{42} & A_{52} \\ 
0 & I_{k_{2}} & 0 & B_{22} & B_{31} & B_{41} \\ 
0 & 0 & I_{k_{3}} & C_{11}+C_{12}+C_{13} & C_{2} & C_{31}+C_{32}+C_{33} \\ 
0 & 0 & 0 & I_{k_{4}} & D_{1} & D_{2}%
\end{array}%
\right] \\
G_{3} &=&\left[ 
\begin{array}{cccccc}
I_{k_{1}} & A_{1} & A_{2} & A_{31} & A_{41} & A_{51}+A_{53} \\ 
0 & 0 & B_{1} & B_{21} & 0 & B_{43} \\ 
0 & 0 & I_{k_{3}} & C_{11}+C_{12}+C_{13} & C_{2} & C_{31}+C_{32}+C_{33} \\ 
0 & 0 & 0 & I_{k_{4}} & D_{1} & D_{2}%
\end{array}%
\right]
\end{eqnarray*}

It is easy to see that $\left\vert C\right\vert =\left\vert C_{1}\right\vert
\left\vert C_{2}\right\vert \left\vert C_{3}\right\vert =2^{3\left(
k_{1}+k_{2}+k_{3}\right)
+2k_{4}+k_{5}}=8^{k_{1}+k_{2}+k_{3}}4^{k_{4}}2^{k_{5}}.$

\begin{corollary}
If $G_{1},G_{2},$and $G_{3}$ are generator matrices of binary linear codes $%
C_{1},C_{2}$ and $C_{3}$ respectively, then the generator matrix of $C$ is
\end{corollary}

$\ \ \ \ \ \ \ \ \ \ \ \ \ \ \ \ \ \ \ \ \ \ \ \ \ \ \ \ \ \ \ \ \ \ \ \ \ \
\ \ \ \ G=\left[ 
\begin{array}{c}
vG_{1} \\ 
\left( 1+v\right) G_{2} \\ 
\left( 1+v^{2}\right) G_{3}%
\end{array}%
\right] $ $\ \ \ \ \ \ \ \ \ \ \ \ \ \ \ \ \ \ \ \ \ \ \ \ \ \ \ \ \ \ \ \ \
\ \ \ \ \ \ \ \ \ \ \ \ \ \ \ \ \ \ \ \ \ \ \ \ \ \ \ \blacksquare \ \ \ \ \
\ \ \ \ \ \ \ \ \ \ \ \ \ \ \ \ \ \ \ \ \ \ \ \ \ \ \ \ \ \ \ \ \ \ \ \ \ \
\ \ \ \ \ $

Let $d_{L}$ minimum Lee weight of linear code $C$ over $R$. Then $d_{L}=\min
\{d_{H}\left( C_{1}\right) ,$ $\ \ \ \ \ \ \ d_{H}\left( C_{2}\right)
,d_{H}\left( C_{3}\right) \}$ where $d_{H}\left( C_{i}\right) $ denotes the
minimum Hamming weights of binary codes $C_{1},C_{2}$ and $C_{3}.$

\section{Quantum Codes From Cyclic Codes Over $R$}

\begin{theorem}
$\left( \text{CSS Construction}\right) $ Let $C$ and $\acute{C}$ be two
binary codes with parameters $\left[ n,k_{1},d_{1}\right] $ and $\left[
n,k_{2},d_{2}\right] $, respectively. If $C^{\perp }\subseteq \acute{C}$,
then an $[[n,k_{1}+k_{2}-n,\min \{d_{1},d_{2}\}]]$ quantum code can be
constructed. Especially, if $C^{\perp }\subseteq C$, then there exists an $%
\left[ \left[ n,2k_{1}-n,d_{1}\right] \right] $ quantum code.
\end{theorem}

\begin{proposition}
Let $C=vC_{1}\oplus \left( 1+v\right) C_{2}\oplus \left( 1+v^{2}\right)
C_{3} $ be a linear code over $R.$Then $C$ is a cyclic code over $R$ iff $%
C_{1},C_{2}$ and $C_{3}$ are binary cyclic codes.
\end{proposition}

\begin{proposition}
Suppose $C=vC_{1}\oplus \left( 1+v\right) C_{2}\oplus \left( 1+v^{2}\right)
C_{3}$ is a cyclic code of length $n$ over $R$. Then $C=\left\langle
vf_{1},\left( 1+v\right) f_{2},\left( 1+v^{2}\right) f_{3}\right\rangle $
and $\left\vert C\right\vert =2^{3n-(\deg f_{1}+\deg f_{2}+\deg f_{3})}$
where $f_{1},f_{2}$ and $f_{3}$ generator polynomials of $C_{1},C_{2}$ and $%
C_{3}$ respectively.
\end{proposition}

\begin{proposition}
Suppose $C$ is a cyclic code of length $n$ over $R$, then there is a unique
polynomial $f\left( x\right) $ such that $C=\left\langle f\left( x\right)
\right\rangle $ and $f\left( x\right) \mid x^{n}-1$ where $f\left( x\right)
=vf_{1}\left( x\right) +\left( 1+v\right) f_{2}\left( x\right) +\left(
1+v^{2}\right) f_{3}\left( x\right) .$
\end{proposition}

\begin{proposition}
If \ $C=vC_{1}\oplus \left( 1+v\right) C_{2}\oplus \left( 1+v^{2}\right)
C_{3}$ is a cyclic code of length $n$ over $R$, then $C^{\perp
}=\left\langle vh_{1}^{\ast }+\left( 1+v\right) h_{2}^{\ast }+\left(
1+v^{2}\right) h_{3}^{\ast }\right\rangle $ and $\left\vert C^{\perp
}\right\vert =2^{\deg f_{1}+\deg f_{2}+\deg f_{3}}$ where for $i=1,2,3$, $%
h_{i}^{\ast }$ are the reciprocal polynomials of $h_{i}$ i.e., $h_{i}\left(
x\right) =\left( x^{n}-1\right) /f_{i}\left( x\right) ,$ $h_{i}^{\ast
}\left( x\right) =x^{\deg h_{i}}h_{i}\left( x^{-1}\right) $ for $i=1,2,3$.
\end{proposition}

\begin{lemma}
A binary linear cyclic code $C$ with generator polynomial $f\left( x\right) $
contains its dual code iff
\end{lemma}

\begin{equation*}
x^{n}-1\equiv 0\left( \func{mod}ff^{\ast }\right)
\end{equation*}

where $f^{\ast }$ is the reciprocal polynomial of $f$.

\begin{theorem}
Let $C=\left\langle vf_{1}+\left( 1+v\right) f_{2}+\left( 1+v^{2}\right)
f_{3}\right\rangle $ be a cyclic code of length $n$ over $R$. Then $C^{\perp
}\subseteq C$ iff $x^{n}-1\equiv 0\left( \func{mod}f_{i}f_{i}^{\ast }\right) 
$ for $i=1,2,3.$
\end{theorem}

\begin{proof}
Let $x^{n}-1\equiv 0\left( \func{mod}f_{i}f_{i}^{\ast }\right) $ for $%
i=1,2,3.$ Then $C_{1}^{\perp }\subseteq C_{1},C_{2}^{\perp }\subseteq
C_{2},C_{3}^{\perp }\subseteq C_{3}.$ By using $vC_{1}^{\perp }\subseteq
vC_{1},$ $\left( 1+v\right) C_{2}^{\perp }\subseteq \left( 1+v\right) C_{2},$
$\left( 1+v^{2}\right) C_{3}^{\perp }\subseteq \left( 1+v^{2}\right) C_{3}.$
We have $vC_{1}^{\perp }\oplus \left( 1+v\right) C_{2}^{\perp }\oplus \left(
1+v^{2}\right) C_{3}^{\perp }\subseteq vC_{1}\oplus \left( 1+v\right)
C_{2}\oplus \left( 1+v^{2}\right) C_{3}.$ So, $\left\langle vh_{1}^{\ast
}+\left( 1+v\right) h_{2}^{\ast }+\left( 1+v^{2}\right) h_{3}^{\ast
}\right\rangle \subseteq \left\langle vf_{1},\left( 1+v\right) f_{2},\left(
1+v^{2}\right) f_{3}\right\rangle .$ That is $C^{\perp }\subseteq C$.

Conversely, if $C^{\perp }\subseteq C$, then $vC_{1}^{\perp }\oplus \left(
1+v\right) C_{2}^{\perp }\oplus \left( 1+v^{2}\right) C_{3}^{\perp
}\subseteq vC_{1}\oplus \left( 1+v\right) C_{2}\oplus \left( 1+v^{2}\right)
C_{3}.$ By thinking $\func{mod}v,\func{mod}\left( 1+v\right) ,\func{mod}%
\left( 1+v^{2}\right) $ respectively we have $C_{i}^{\perp }\subseteq C_{i}$
for $i=1,2,3$. Therefore, $x^{n}-1\equiv 0\left( \func{mod}f_{i}f_{i}^{\ast
}\right) $ for $i=1,2,3.$
\end{proof}

\begin{corollary}
$C=vC_{1}\oplus \left( 1+v\right) C_{2}\oplus \left( 1+v^{2}\right) C_{3}$
is a cyclic code of length $n$ over $R$. Then $C^{\perp }\subseteq C$ iff $%
C_{i}^{\perp }\subseteq C_{i}$ for $i=1,2,3$.
\end{corollary}

\begin{example}
Let $\ n=21,R=F_{2}+vF_{2}+v^{2}F_{2}$%
\begin{eqnarray*}
x^{21}-1 &=&(x+1)(x^{2}+x+1)(x^{3}+x^{2}+1)(x^{3}+x+1)(x^{6}+x^{4}+x^{2}+x+1)
\\
&&(x^{6}+x^{5}+x^{4}+x^{2}+1)
\end{eqnarray*}
\end{example}

in $F_{2}\left[ x\right] .$

Let $f_{1}=x+1,$ $f_{2}=x^{2}+x+1,$ $f_{3}=x^{3}+x^{2}+1,$ $f_{4}=x^{3}+x+1,$
$f_{5}=x^{6}+x^{4}+x^{2}+x+1,$ $f_{6}=x^{6}+x^{5}+x^{4}+x^{2}+1$. Hence,

$f_{1}^{\ast }=x+1=$ $f_{1},$ $f_{2}^{\ast }=x^{2}+x+1=f_{2},$ $f_{3}^{\ast
}=x^{3}+x+1=f_{4},$ $f_{4}^{\ast }=x^{3}+x^{2}+1=$ $f_{3},$ $f_{5}^{\ast
}=x^{6}+x^{5}+x^{4}+x^{2}+1=$ $f_{6},$ $f_{6}^{\ast }=x^{6}+x^{4}+x^{2}+x+1=$
$f_{5}.$

Let $C=\left\langle vf_{i},\left( 1+v\right) f_{i},\left( 1+v^{2}\right)
f_{i}\right\rangle $ for $i=3,4,5,6$. Obviously $x^{n}-1$ is divisibly by $%
f_{i}f_{i}^{\ast }$, for $i$. Thus we have $C^{\perp }\subseteq C.$

Using Theorem $9$ and Theorem $15$ we can construct quantum codes.

\begin{theorem}
Let $C=vC_{1}\oplus \left( 1+v\right) C_{2}\oplus \left( 1+v^{2}\right)
C_{3} $ be a cyclic code of arbitrary length $n$ over $R$ with type $%
8^{k_{1}+k_{2}+k_{3}}4^{k_{4}}2^{k_{5}}.$ If $C_{i}^{\perp }\subseteq C_{i}$
where $i=1,2,3$ then $C^{\perp }\subseteq C$ and there exists a quantum
error-correcting code with parameters $\left[ \left[ 3n,6\left(
k_{1}+k_{2}+k_{3}\right) +4k_{4}+2k_{5}-3n,d_{L}\right] \right] $ where $%
d_{L}$ is the minimum Lee weights of $C.$
\end{theorem}

\section{Examples}

\begin{example}
Let $n=8$
\end{example}

\begin{equation*}
x^{8}-1=\left( x+1\right) \left( x+1\right) \left( x+1\right) \left(
x+1\right) \left( x+1\right) \left( x+1\right) \left( x+1\right) \left(
x+1\right)
\end{equation*}%
in $F_{2}\left[ x\right] .$ Let $f_{1}\left( x\right) =f_{2}\left( x\right)
=f_{3}\left( x\right) =x^{3}+x^{2}+x+1$.

Thus $C=\left\langle vf_{1},\left( 1+v\right) f_{2},\left( 1+v^{2}\right)
f_{3}\right\rangle =\left\langle v^{2}\left( x^{3}+x^{2}+x+1\right)
\right\rangle $. $C$ is a linear code of length $8$ and minimum Lee weight $%
d_{L}=2$. The dual code $C^{\perp }=\left\langle vh_{1}^{\ast },\left(
1+v\right) h_{2}^{\ast },\left( 1+v^{2}\right) h_{3}^{\ast }\right\rangle
=\left\langle v^{2}\left( x^{5}+x^{4}+x+1\right) \right\rangle $ can be
obtained. Clearly, $C^{\perp }\subseteq C$. Hence we obtain a quantum code
with parameters $\left[ \left[ 24,6,2\right] \right] .$

\begin{example}
Let $n=8$
\end{example}

\begin{equation*}
x^{8}-1=\left( x+1\right) \left( x+1\right) \left( x+1\right) \left(
x+1\right) \left( x+1\right) \left( x+1\right) \left( x+1\right) \left(
x+1\right)
\end{equation*}%
in $F_{2}\left[ x\right] .$ Let $f_{1}\left( x\right) =f_{2}\left( x\right)
=f_{3}\left( x\right) =x^{2}+1$.

Thus $C=\left\langle vf_{1},\left( 1+v\right) f_{2},\left( 1+v^{2}\right)
f_{3}\right\rangle =\left\langle v^{2}\left( x^{2}+1\right) \right\rangle $. 
$C$ is a linear code of length $8$ and minimum Lee weight $d_{L}=2$.

The dual code $C^{\perp }=\left\langle vh_{1}^{\ast },\left( 1+v\right)
h_{2}^{\ast },\left( 1+v^{2}\right) h_{3}^{\ast }\right\rangle =\left\langle
v^{2}\left( x^{6}+x^{4}+x^{2}+1\right) \right\rangle $ can be obtained.
Clearly, $C^{\perp }\subseteq C.$ Hence we obtain a quantum code with
parameters $\left[ \left[ 24,12,2\right] \right] .$ Let $f_{1}\left(
x\right) =f_{2}\left( x\right) =f_{3}\left( x\right) =x+1.$Thus we obtain a
quantum code with parameters $\left[ \left[ 24,18,2\right] \right] $.

\begin{example}
Let $n=7$
\end{example}

\begin{equation*}
x^{7}-1=\left( x+1\right) \left( x^{3}+x+1\right) \left( x^{3}+x^{2}+1\right)
\end{equation*}%
in $F_{2}\left[ x\right] .$ Let $f_{1}\left( x\right) =f_{2}\left( x\right)
=f_{3}\left( x\right) =x^{3}+x+1$.

Thus $C=\left\langle vf_{1},\left( 1+v\right) f_{2},\left( 1+v^{2}\right)
f_{3}\right\rangle =\left\langle v^{2}\left( x^{3}+x+1\right) \right\rangle $%
. $C$ is a linear code of length $7$ and minimum Lee weight $d_{L}=3.$The
dual code $C^{\perp }=\left\langle vh_{1}^{\ast },\left( 1+v\right)
h_{2}^{\ast },\left( 1+v^{2}\right) h_{3}^{\ast }\right\rangle =\left\langle
v^{2}\left( x^{4}+x^{2}+x+1\right) \right\rangle $ can be obtained. Clearly, 
$C^{\perp }\subseteq C.$ Hence we obtain a quantum code with parameters $%
\left[ \left[ 21,3,3\right] \right] .$

\begin{example}
Let $n=15$
\end{example}

\begin{equation*}
x^{15}-1=\left( x+1\right) \left( x^{2}+x+1\right) \left( x^{4}+x+1\right)
\left( x^{4}+x^{3}+1\right) \left( x^{4}+x^{3}+x^{2}+x+1\right)
\end{equation*}%
in $F_{2}\left[ x\right] .$ Let $f_{1}\left( x\right) =f_{2}\left( x\right)
=f_{3}\left( x\right) =x^{4}+x+1$.

Thus $C=\left\langle vf_{1},\left( 1+v\right) f_{2},\left( 1+v^{2}\right)
f_{3}\right\rangle =\left\langle v^{2}\left( x^{4}+x+1\right) \right\rangle $%
. $C$ is a linear code of length $15$ and minimum Lee weight $d_{L}=3.$The
dual code $C^{\perp }=\left\langle vh_{1}^{\ast },\left( 1+v\right)
h_{2}^{\ast },\left( 1+v^{2}\right) h_{3}^{\ast }\right\rangle $ can be
obtained. Clearly, $C^{\perp }\subseteq C.$ Hence we obtain a quantum code
with parameters $\left[ \left[ 45,21,3\right] \right] .$

\begin{example}
Let $n=16$
\end{example}

\begin{equation*}
x^{16}-1=\left( x+1\right) ^{16}
\end{equation*}%
in $F_{2}\left[ x\right] .$ Let $f_{1}\left( x\right) =f_{2}\left( x\right)
=f_{3}\left( x\right) =x^{3}+x^{2}+x+1$.

Thus $C=\left\langle vf_{1},\left( 1+v\right) f_{2},\left( 1+v^{2}\right)
f_{3}\right\rangle =\left\langle v^{2}\left( x^{3}+x^{2}+x+1\right)
\right\rangle $. $C$ is a linear code of length $16$ and minimum Lee weight $%
d_{L}=2.$The dual code $C^{\perp }=\left\langle vh_{1}^{\ast },\left(
1+v\right) h_{2}^{\ast },\left( 1+v^{2}\right) h_{3}^{\ast }\right\rangle
=\left\langle v^{2}\left( x^{13}+x^{12}+x^{9}+x^{8}+x^{5}+x^{4}+x+1\right)
\right\rangle $ can be obtained. Clearly, $C^{\perp }\subseteq C.$ Hence we
obtain a quantum code with parameters $\left[ \left[ 48,30,2\right] \right]
. $ Let $f_{1}\left( x\right) =f_{2}\left( x\right) =f_{3}\left( x\right)
=x^{4}+1$. Thus we obtain a quantum code with parameters $\left[ \left[
48,24,2\right] \right] .$

\begin{example}
Let $n=21,$
\end{example}

\begin{eqnarray*}
x^{21}-1 &=&(x+1)(x^{2}+x+1)(x^{3}+x^{2}+1)(x^{3}+x+1)(x^{6}+x^{4}+x^{2}+x+1)
\\
&&(x^{6}+x^{5}+x^{4}+x^{2}+1)
\end{eqnarray*}%
in $F_{2}\left[ x\right] .$ Let $f_{1}\left( x\right) =f_{2}\left( x\right)
=f_{3}\left( x\right) =x^{6}+x^{5}+x^{4}+x^{2}+1$.

Thus $C=\left\langle vf_{1},\left( 1+v\right) f_{2},\left( 1+v^{2}\right)
f_{3}\right\rangle =\left\langle v^{2}\left( x^{4}+x+1\right) \right\rangle $%
. $C$ is a linear code of length $21$ and minimum Lee weight $d_{L}=3.$The
dual code $C^{\perp }=\left\langle vh_{1}^{\ast },\left( 1+v\right)
h_{2}^{\ast },\left( 1+v^{2}\right) h_{3}^{\ast }\right\rangle $ can be
obtained. Clearly, $C^{\perp }\subseteq C.$ Hence we obtain a quantum code
with parameters $\left[ \left[ 63,27,3\right] \right] .$ Let $f_{1}\left(
x\right) =f_{2}\left( x\right) =f_{3}\left( x\right) =x^{3}+x^{2}+1$. Thus
we obtain a quantum code with parameters $\left[ \left[ 63,45,3\right] %
\right] .$

\textbf{Acknowledgement}

We would like to thank S.T.Dougherty for his many suggestions which helped
us.

\section{References}

$\ \ \ \left[ 1\right] $\ A.M.Steane,Simple quantum error correcting codes,
Phys. Rev. A, $54\left( 1996\right) ,$ $4741-4751.$

$\left[ 2\right] $ A.R.Calderbank, E.M.Rains, P.M.Shor, N.J.A.Sloane,
Quantum error correction via codes over $GF\left( 4\right) ,$ IEEE Trans.
Inf. Theory, $44\left( 1998\right) ,$ $1369-1387.$

$\left[ 3\right] $ J.Qian, Quantum codes from cyclic codes over $%
F_{2}+vF_{2},$ Journal of Inform.$\&$ computational Science $10:6\left(
2013\right) ,$ $1715-1722.$

$\left[ 4\right] $ J.Qian, W.Ma, W.Gou, Quantum codes from cyclic codes over
finite ring, Int. J. Quantum Inform., $7\left( 2009\right) ,$ $1277-1283.$

$\left[ 5\right] $ P.W.Shor,Scheme for reducing decoherence in quantum
memory, Phys. Rev. A, $52\left( 1995\right) ,$ $2493-2496.$

$\left[ 6\right] $ S.Zhu, Y.Wang, M.Shi, Some Results On Cyclic Codes Over $%
F_{2}+vF_{2},$ IEEE Transactions on Information Theory,vol.$56$, No.$4$, $%
2010$

$\left[ 7\right] $ X.Kai,S.Zhu, Quaternary construction bof quantum codes
from cyclic codes over $F_{4}+uF_{4},$ Int. J. Quantum Inform., $9\left(
2011\right) ,$ $689-700.$

$\left[ 8\right] $ X.Xu, S.Yun, Y.Xiong, Generator Matrix of The Linear
Codes And Gray Images Over Ring $F_{2}+vF_{2}+v^{2}F_{2},$ Journal of
Applied Sciences $13\left( 4\right) \left( 2013\right) :650-653.$

$\left[ 9\right] $ X.Yin, W.Ma, Gray Map And Quantum Codes Over The Ring $%
F_{2}+uF_{2}+u^{2}F_{2},$ International Joint Conferences of IEEE
TrustCom-11,$2011.$

\end{document}